\documentclass[12pt,a4paper,normalheadings,headsepline,headinclude,bibtotoc,fleqn,review]{article}
\usepackage[latin1]{inputenc}
\usepackage{amsmath,array,amsthm,graphicx,amssymb,natbib,color,mathtools,xspace,enumitem}
\usepackage[lines=10]{geometry}

\newtheorem{prop}{Proposition}

\usepackage{fancyhdr} 
\begin{document}

\begin{center}\huge{Darwinian Adverse Selection}\footnote{I thank Martin Hellwig, George Loewenstein, Michael
Mandler, Carl Christian von Weizsäcker, Martin Salm, and
participants of the 2014 MPI seminar in Bonn for comments and
questions.}
\end{center}
\begin{center}\emph{Wolfgang Kuhle}\\\emph{Max Planck Institute for Research on
Collective Goods, Kurt-Schumacher-Str. 10, 53113 Bonn, Germany.
Email: kuhle@coll.mpg.de.}
\end{center}

\noindent\emph{\textbf{Abstract:} We develop a model to study the
role of rationality in economics and biology. The model's agents
differ continuously in their ability to make rational choices. The
agents' objective is to ensure their individual survival over time
or, equivalently, to maximize profits. In equilibrium, however,
rational agents who maximize their objective survival probability
are, individually and collectively, eliminated by the forces of
competition. Instead of rationality, there emerges a unique
distribution of irrational players who are individually not fit
for the struggle of survival. The selection of irrational players
over rational ones relies on the fact that all rational players
coordinate on the same optimal action, which leaves them collectively undiversified and thus vulnerable to aggregate risks.}\\
\textbf{Keywords: Maximization, Rationality, Economics, Biology, Group Selection}\\
\textbf{JEL: D81, D01, D03, G11}


\vspace{.2cm}


\section{Introduction}\label{intro}

In economics it is commonplace that ``man's ability to operate as
a logical animal capable of systematic empirical induction was
itself the direct outcome of the Darwinian struggle for
survival".\footnote{\citet{Sam72}, p. 249. See \citet{Fri53},
p.22, and \citet{Sch33} for similar arguments. See \citet{Nel85}
and \citet{Win71} for more differentiated references.} Most
economic models therefore assume that the forces of competition
ensure the elimination of agents who are not capable of
profit-maximizing behavior. One advantage of such rational
maximizing behavior is that it is complementary to mathematical
optimization techniques which accommodate systematic model
building. Moreover, it follows from the laws of diminishing
marginal returns that most maximization problems feature unique
solutions. Accordingly, rational choice yields unique predictions
which serve as a natural reference point.

For the present purpose, we summarize the case for rational choice
models as follows: (i) \emph{The competitive struggle for survival
selects rational agents over those who make mistakes.} Moreover,
(ii) \emph{unlike rational maximizing behavior, which is well
defined in the context of mathematical models, irrational
``choice" is inherently hard to define and thus cannot serve as a
benchmark.}

In this paper, we develop a simple dynamic neoclassical model of
``Darwinian adverse selection''. While the model's assumptions are
neoclassical, its theoretical predictions invert our previous
theses (i) and (ii). The model economy is inhabited by agents who
differ continuously in their ability to make rational choices and
the agents' objective is to ensure their individual survival over
time or, equivalently, to maximize profits. And, indeed, the
rational choice by which individuals maximize their individual
survival probability will be uniquely defined. The fact that
maximizing behavior is well defined, however, has the implication
that rational players are eliminated by the forces of competition,
while irrational player
types, not capable of maximizing, survive.
We derive this result for two standard economic model
environments.

The key intuition is very simple: Suppose agents have to choose
between two ferry boats. One is well maintained, the other is
leaking. All rational agents will therefore board the
well-maintained ferry, while irrational agents, not capable of
maximizing their individual survival probability, may be found on
either boat. Consequently, to ensure that some irrational players
survive, it suffices that one of the two boats does not sink.
Rational players, on the contrary, individually and collectively
depend on only one boat. And thus the probability of some rational
players surviving is strictly lower than the probability of some
irrational players surviving. Put differently, individual
rationality has the consequence that it aligns the choices of
rational players. Accordingly, rational players as a group suffer
from under-diversification.

Compared to our initial theses (i) and (ii), we therefore find
that our model yields rather different results. Rationality with
its unique predictions is the very reason for the extinction of
rational players. Irrational agents, on the contrary, survive not
despite, but because of their individual mistakes. Accordingly, in
relation to our first thesis (i), we derive a simple antithesis
(i'): \emph{The struggle for survival selects irrational agents,
who make mistakes over those who make no mistakes in the struggle
for survival.} Moreover, irrationality serves as the reference
point: there exists only one unique distribution of irrational
behavior, which is well defined in the context of analytical model
building.

In Section \ref{s1}, we introduce the model and present the main
result. In our baseline model, players are faced with a simple
binary choice problem. In Section \ref{s2}, we drive our results
for a capital market setting where agents can choose from an
arbitrary number of alternatives. Finally, we discuss assumptions,
compare our results to the literature on group selection, and
suggest different interpretations. Section \ref{s3} concludes.


\section{Model}\label{s1}

At the beginning of time, entrepreneurs choose between two
actions/``production technologies" $A$ and $B$. After this choice
is made, nature randomly selects between two states $0$ and $1$.
State $0$ is selected with probability $P>1/2$. State $1$ is less
likely and occurs with the residual probability $(1-P)<1/2$. If
state $0$ $(1)$ is selected, then technology $A$ $(B)$ is more
productive than technology $B$ $(A)$. Accordingly, if state $0$ is
selected, the competitive market system eliminates those firms
that chose the inefficient technology $B$.\footnote{An economy
consistent with such an assumption would, for example, involve the
assumption that (i) goods can be sold at price $p\geq1$, (ii) each
entrepreneur has one unit of labor $L$ and needs one ``unit" of
income for subsistence/survival, and (iii) production is
$Y=\tau_{i}L, i=A,B$ with $\tau_A=1,\tau_B=0$, in state $0$ and
$\tau_A=0,\tau_B=1$ in state $1$. This environment can be
generalized incorporating elastic demand, varying returns to
scale, and production involving several inputs.} Likewise in state
$1$, where players with technology $B$ drive those out of the
market who chose $A$. Agents who used the wrong technology are
eliminated. Successful entrepreneurs move on to the next period
where they have to choose once again between two technologies $A$
and $B$. Time $t$ goes from $t=0,1,2...,\infty$.

Entrepreneurs differ in types $\phi\in[0,1]$. These types vary
continuously from rational to completely irrational. More
precisely, an agent $i$ who is of type $\phi$ plays the following
strategies:
$$S_i=\left\{\begin{array}{ll}A & \qquad\text{with probability } \phi\\
B & \qquad\text{with probability } 1-\phi \, .
\end{array}\right.$$

In the following, we identify players' rationality by their innate
probability, $\phi\in[0,1]$, to play the dominant strategy $A$,
which maximizes the individual survival probability.\footnote{That
is, choosing $A$ maximizes both the probability of surviving to
the next period as well as the probability of making a positive
profit since the probability $P>1/2$ with which $A$ is the
superior technology exceeds the probability $1-P<1/2$ with which
$B$ is the superior technology.} Individual profit maximization
therefore requires agents to play $A$ all the time and hence
$\phi=1$ agents are rational. As $\phi$ decreases, agents put less
and less emphasis on the dominant action and eventually agents
with $\phi=0$ reliably play $B$, which minimizes their survival
probability. At the beginning of time, there exists a density one
of each player type $\phi$. Moreover, players of type $\phi$
randomize independently over the two actions $A$ and $B$ such that
we can use the weak law of large numbers to calculate the share of
type $\phi$ agents who play $A$ and $B$ as $\phi$ and $1-\phi$,
respectively.

\begin{prop}\label{p1} Over time, rational player types $\phi=1$ who maximize their individual survival probability, are eliminated
with probability one. All agents $\phi=0$ who minimize their
probability of survival are also eliminated. The probability of
some agents of all other types $\phi\in(0,1)$, who play strategies
that are neither consistent with maximization nor minimization,
surviving to any point in time is equal to one.
\end{prop}
\begin{proof} Rational agents $(\phi=1)$, who maximize their individual survival probability,
choose strategy $A$, with its superior survival probability
$P>1/2$, all the time. Hence, $P^t$ is the probability of a
rational agent surviving to period $t$. Indeed rationality iduces
all rational agents to play $A$ such that the \emph{collective
survival probability} for the entire group of rational agents
coincides with the individual survival probability $P$. The
long-run survival probability of rational agents, individually and
collectively, is therefore $\lim_{t\rightarrow\infty}P^t=0$.
Second, the individual and collective survival probability for
minimizing agents $(\phi=0)$ is $(1-P)$. Hence, over time, we have
$\lim_{t\rightarrow\infty}(1-P)^t=0$. Finally, irrational agents
$\phi\in(0,1)$, who do not maximize: The individual survival
probability of a type $\phi$ agent is given by the convex
combination $P\phi+(1-P)(1-\phi)\in[1-P,P]$. That is, as
individuals, irrational agents survive with a probability that
falls short of the rational agent's probability $P$. However, a
share $\phi$ of irrational type $\phi$ players survives in state
$0$, and a share $1-\phi$ in state $1$. The probability of some
type $\phi$ players surviving from one period to the next is
therefore $P+(1-P)=1$.
\end{proof}

Proposition \ref{p1} reflects that the extinction of rational
agents is an immediate consequence of the very fact that rational
agents maximize their individual survival probabilities. In doing
so, they coordinate on the one rational answer that exists to the
struggle for survival that nature presents them with. However,
since individual rationality is well-defined in the present model,
rational players all take the same bet and are thus eliminated
collectively in the rare event where technology $B$ dominates $A$.
Put differently, nature rewards a diversification, which is
individually irrational.

\section{Discussion}\label{s2}

The previous result suggests that the failure of rationality
originates from a lack of diversification. In this section, we
reflect on this finding. At first sight, the binary choice model
from the previous section raises the question whether rational
agents might find some way to diversify once they can choose from
more than just two actions. To address this concern, we show that
our findings also obtain in a standard capital market environment
where rational agents can choose from an arbitrary number of
different assets to build diversified portfolios. Second, we
emphasize the fact that the present model implies the emergence of
equilibrium biases. Third, we relate our conclusions to the
literature on individual and group fitness.

\subsection{Diversification}

Proposition \ref{p1} clearly shows that the survival of irrational
players relies on the fact that rational players indeed perform to
the best of their ability. Accordingly, for every given task, they
come to the same optimal conclusion. And all rational players will
choose the \emph{same} action $A$. Everyone choosing $A$, however,
implies that rational players' choices are \emph{collectively}
perfectly correlated. This aligns their actions perfectly, which
in the presence of aggregate risks, means that they perish
simultaneously. Hence, rational players, \emph{as a collective},
do not diversify. To emphasize this point, we develop a simple
version of the workhorse capital asset pricing model (CAPM) in
Appendix \ref{A1}. In this model, investors maximize the
probability with which their individual asset returns $Y_j$ do not
fall below a minimum $Y^{min}$. As \citet{Sha91} shows, the key
prediction, namely that all rational investors choose the same
``market portfolio", holds under much more general circumstances.
That is, even though investors can choose from infinitely many
different assets, they all choose the same diversified ``market
portfolio", which we call $A$. This portfolio indeed maximizes
their individual survival probability. Moreover, since all
rational agents choose the same optimal portfolio $A$, their
\emph{collective survival probability} is identical to their
individual survival probability. Irrational players play $B$. That
is, they do not diversify and select one individual asset at
random. Hence, their individual survival probability falls short
of that of the rational players. However, as a collective, they
are fully diversified: Regardless of how deep the over all market
falls, a positive measure of irrational agents will survive. Put
differently, it turns out that this model's predictions are in
line with the simple binary choice model from Section \ref{s1}. In
Appendix \ref{A1}, we substantiate the foregoing claims and prove

\begin{prop}\label{p2} All rational investors who maximize their survival probability choose the market portfolio
$A$. And, over time, the probability of rational investors' funds
being closed due to catastrophic losses goes to $1$. The
probability with which a positive measure of irrational investors,
who individually invest all of their wealth into one randomly
drawn asset, survive to any point in time is one.\end{prop}

Put differently, those investors who behave according to the
predictions of the capital asset pricing model will eventually go
bankrupt in one unlikely market downturn. In this downturn,
however, there are always some outlier assets, and those
irrational investors who chose these outlier assets survive the
downturn.

\subsubsection{Preferences and additional strategies}

From the previous analysis it is clear that the population of
players declines over time. Moreover, there is no safe act that
would allow agents to survive until the next period with
probability one. If we were to introduce safe acts alongside
rewards to risk-taking, which would come in the form of increased
fertility in case the risky strategy pays off, we would obtain the
same results as before if rational players maximize the number of
expected offspring by playing one risky act, ``$A$", all the time.
If rational players were to maximize their survival probability
instead, they would eventually be marginalized in a population
where other player types individually randomize over the safe and
the risky option where expected fertility is higher. Once again,
the key to this finding would lie in the observation that rational
players, regardless of what their aim is, would not have an
incentive to randomize individually. Finally, we note that the
same arguments would apply if there were two groups of rational
players, one maximizing the expected number of children, the other
maximizing the probability of survival.



\subsection{Trembling Hand and Heterogenous Priors}

The foregoing model has shown that an endogenous population of
players evolves over time. And the probability of rational players
being included in this distribution goes to zero as time evolves,
leaving only agents playing trembling hand. The present model thus
provides a framework where the often criticized trembling hand
assumption of \citet{Sel75} is a model outcome rather than an
assumption.

In an alternative interpretation, the present model can explain
the emergence of heterogenous priors. Suppose that all players die
after one period and only those who made the right choice have one
``child". In turn, the child of a type $\phi$ agent will, with
probability $\phi$ (respectively $1-\phi$), hold the prior belief
that the probability $P$ with which action $A$ dominates $B$ is
$P>1/2$ $(P<1/2)$. By the law of large numbers, we would have a
stable ``sex ratio" and a share $\phi$ play $A$ while a share
$1-\phi$ play $B$. For each history of events, agents will
therefore know that they live in a society that agrees to
disagree. In this interpretation, the present model explains the
empirical fact that agents ``agree to disagree", which - as
\citet{Aum76} points out - is otherwise hard to justify.

\subsection{Group Selection and Absolute Individual Fitness}

In a different interpretation, the present model supports the
concept of group selection rather than selection based on
individual absolute fitness. That is, over time, rational agents,
who are individually the fittest, perish collectively with
probability one. On the contrary, irrational agents, who are
individually not capable of maximizing their survival
probabilities, prevail collectively with probability one. At the
same time, the present model does not rely on any of the
traditional arguments for group selection. That is, the model
involves no public good argument. Moreover, agents do not (and
cannot) act on any altruistic motives; choosing the dominated
action $B$ does not increase the survival probability of any of
the other players.\footnote{\citet{Wil04,Wil75,Wil07},
\citet{Que92}, \citet{Sam93} and \citet{Eld08} for the role of
public goods, (implicit) cooperation, and altruism in the context
of group selection models.} Nor does the model involve arguments
concerning the wastes of competition that would apply if, for
example, one agent's choice of $A$ lowered the other agent's
survival probabilities once they also play $A$. And there are no
rewards for risk-taking. The argument is not based on risk
preferences or strategic considerations either, such as prisoners
dilemmas or hawk-dove tradeoffs where, even though hawks can
exploit doves, not all agents choose to be hawks since hawk-hawk
encounters can be deadly.\footnote{\citet{Wol07}.} Finally,
rationality is costless. Instead of these previous arguments, the
present selection argument relies on the fact that nature rewards
a diversification which rationality, with its unique predictions,
cannot deliver.



\section{Conclusion}\label{s3}

Individual rationality with its unambiguous predictions makes
model building operational. Moreover, since rational agents
outperform irrational ones, the forces of competition will
eliminate irrational agents over time. In the current model,
rational players indeed perform to the best of their ability.
Accordingly, for every given task, they come to the same optimal
conclusion. As discussed, if $A$ is the well diversified ``market
portfolio" and $B$ an inferior, badly diversified portfolio, then
all rational players choose the \emph{same} diversified market
portfolio $A$. Everyone holding $A$, however, implies that
rational players' choices are \emph{collectively} perfectly
correlated. This aligns their actions perfectly, which means in
the presence of aggregate risks, that they perish simultaneously.
Hence, rational players, \emph{as a collective}, do not diversify.
Accordingly, over time, they will eventually be washed away by one
unforseen regime change.

In our model, irrational agents simply invest all their wealth
into one randomly chosen asset. This strategy minimizes their
individual probability of success. At the same time, this behavior
ensures that some irrational agents survive in all states of
nature. Hence, even though irrational players may be outperformed
by their rational counterparts for long periods of time, they
survive those unlikely events where uncomprehensible choices pay
off.


\newpage

\begin{appendix}
\section{Proof of Proposition \ref{p2}}\label{A1}
In this appendix, we first derive the main prediction of the CAPM,
namely that all rational players will hold the same ``market
portfolio", which we call $A$. This portfolio minimizes the
probability with which investors suffer a ``catastrophic loss". As
in the baseline model, investors suffering catastrophic losses are
eliminated from the capital market. Second, we show that rational
investors who diversify optimally have a ``survival probability"
that exceeds the survival probability of irrational investors who
do not diversify. Finally, we show that a positive share of
irrational investors survives even the worst market downturns with
probability one. Rational investors on the contrary are eliminated
in severe downturns.

\textbf{Market:} There exists a continuum of assets:
\begin{eqnarray} y_i=\theta+\xi_i,\quad i\in[0,1],\quad \theta\sim\mathcal{N}(\mu,\sigma^2),  \label{a1}\end{eqnarray}
where $\xi_i$ is i.i.d. white noise
$\xi_i\sim\mathcal{N}(0,\sigma_{\xi}^2)$. Clearly $\theta$
represents the general ``market risk", which is common to all
assets, and $\xi_i$
is the idiosyncratic risk associated with a particular asset $i$. 
In the following, each agent $j$ can choose a portfolio to
minimize the probability that returns $Y_j$ fall short of a
minimum requirement $Y^{min}<\mu$. Investors receiving
$Y_j<Y^{min}$ go bankrupt. Respectively, if we think of the
investor as a fund manager, the fund is closed due to poor
performance if $Y_j<Y^{min}$. Investors receiving $Y_j>Y^{min}$
move on to the next period where the game is repeated... To
substantiate the claims from the main text, we have to show that
there exists one unique optimal ``market portfolio" $A$, which
maximizes the objective survival probability $P(Y_A>Y^{min})$.
Finally, without loss of generality, we normalize all assets'
prices to one, and each investor can invest one unit of currency.
To derive the optimal portfolio, we note that the investor $j$ can
buy the following portfolios:
\begin{eqnarray} &&Y_{j1}=y_1,\quad Y_{j1}\sim\mathcal{N}(\mu,\sigma^2+\sigma_{\xi}^2)\label{1}\\
&&Y_{j2}=\frac{1}{2}y_1+\frac{1}{2}y_2,\quad Y_{j2}\sim\mathcal{N}(\mu,\sigma^2+\frac{1}{2}\sigma_{\xi}^2)\\
&&Y_{j3}=\frac{1}{3}y_1+\frac{1}{3}y_2+\frac{1}{3}y_3,\quad
Y_{j3}\sim\mathcal{N}(\mu,\sigma^2+\frac{1}{3}\sigma_{\xi}^2)\\
&&Y_{jN}=\sum_{n=1}^{N}\frac{1}{N}y_n,\quad
Y_{jN}\sim\mathcal{N}(\mu,\sigma+\frac{1}{N}\sigma_{\xi}^2).
\label{2}\end{eqnarray} \textbf{Rational Choice:} It follows from
(\ref{1})-(\ref{2}) that
$P(Y_{jN}>Y^{min})=\Phi\Big(\sqrt{\frac{1}{\sigma^2+\frac{1}{N}\sigma^2_{\xi}}}(\mu-Y^{min})\Big)$,
where $\Phi()$ is the cumulative density function of the standard
normal distribution. Thus, the survival probability of a fund
manager is monotonously increasing in $N$ since we assumed that
$\mu>Y^{min}$. Accordingly, rational investors will choose
$N=\infty$ to maximize their survival probability. That is, they
include a small amount of every asset $x_i, i\in[0,1]$ in their
portfolio to achieve maximum diversification. This means that
\emph{all rational fund managers} choose the same ``market
portfolio" as predicted by the standard capital asset pricing
model. Once we call this market portfolio portfolio $A$, we have
substantiated the claim from the main text that all rational
managers choose $A$, which puts their individual and collective
survival probability to
$P_A=\Phi\Big(\sqrt{\sigma^2}(\mu-Y^{min})\Big)$.

\textbf{Irrational Agents:} As before, there exists a measure one
of irrational players who do not diversify. They simply invest
their total wealth into one asset which they pick at random from
the continuum of assets $j\in[0,1]$. We call this choice $B$.
Accordingly, the survival probability of an individual irrational
manager is
$P(Y_j>Y^{min})=\Phi\Big(\sqrt{\frac{1}{\sigma^2+\sigma^2_{\xi}}}(\mu-Y^{min})\Big)<P_A$.
The \emph{collective} survival probability of a mass one of
irrational players, however, is again equal to one. That is, for
every given draw $\theta$, we have the distribution of individual
asset returns $y_i|\theta\sim\mathcal{N}(\theta,\sigma_{\xi}^2)$,
and there is always a positive mass
$P(y>Y^{min}|\theta)=\Phi\Big(\sqrt{\frac{1}{\sigma_{\xi}^2}}(\theta-Y^{min})\Big)>0$
of irrational players, who receive a return in excess of the
required minimum.

Put differently, regardless of how deep the general ``market",
$\theta$, falls, there are always some outlier assets, which
deliver a return sufficient to ensure that $Y_i=x_i>Y^{min}$.
Hence, there are always surviving irrational agents. The same is
not true for rational agents, who hold the market portfolio $A$,
which yields
$Y_A=\int_{[0,1]}x_idi=\theta+\int_{[0,1]}\xi_idi=\theta$. Or, in
terms of (\ref{2}), we have $\lim_{N\rightarrow\infty}Y_N=\theta$.
That is, once $\theta<Y^{min}$, which happens with probability
$1-P_A>0$, all rational agents perish simultaneously.

\end{appendix}

\begin{appendix}
\addcontentsline{toc}{section}{References}
\markboth{References}{References}
\bibliographystyle{apalike}
\bibliography{References}

\end{appendix}

\end{document}